\documentclass[journal,onecolumn]{IEEEtran}

\usepackage{tikz}

\usepackage{url}
\usepackage{cancel}
\usepackage{blindtext}
\usepackage{graphicx}
\usepackage{pgfplots}
\usepackage{hyperref}

\usetikzlibrary{through,calc}

\definecolor{bblue}{HTML}{4F81BD}
\definecolor{rred}{HTML}{C0504D}
\definecolor{ggreen}{HTML}{9BBB59}
\definecolor{ppurple}{HTML}{9F4C7C}

\definecolor{mDarkBrown}{HTML}{604c38}
\definecolor{mDarkTeal}{HTML}{23373b}
\definecolor{mMediumTeal}{HTML}{205A65}

\definecolor{mLightTeal}{HTML}{41a3c3}

\definecolor{mLightBrown}{HTML}{EB811B}
\definecolor{mMediumBrown}{HTML}{C87A2F}

\def\axesradargrid at (#1,#2){\draw[help lines,step=0.5] (0,0) grid (#1,#2);
\draw[->, thin](0,#2/2)--(#1,#2/2);
\node [right] at (#1,#2/2) {$\mathcal{R}$};
\draw[->, thin](#1/2,0)--(#1/2,#2);
\node [above] at (#1/2,#2) {$\mathcal{I}$};}

\def\axesradar at (#1,#2){
\draw[->, thin](0,#2/2)--(#1,#2/2);
\node [right] at (#1,#2/2) {$\mathcal{R}$};
\draw[->, thin](#1/2,0)--(#1/2,#2);
\node [above] at (#1/2,#2) {$\mathcal{I}$};}

\def\pointS at (#1:#2){
\draw [mDarkTeal,  thin,fill=mMediumTeal] (#1:#2) circle [radius=0.04];
\draw [mMediumBrown, thin,dashed] (#1:#2) circle [radius=0.2];}

\usepackage{color,colortbl}
\usepackage{cite}

\usepackage{xspace}
\usepackage{graphicx}
\usepackage{epstopdf}
\usepackage{amssymb}
\usepackage[cmex10]{amsmath}

\newcommand{\bs}{\boldsymbol}
\newcommand{\bb}{\mathbb}
\newcommand{\cl}{\mathcal}
\newcommand{\ts}{\textstyle}

\usepackage{mathtools}
\mathtoolsset{showonlyrefs,showmanualtags}

\DeclareMathOperator{\sign}{sign}
\DeclareMathOperator{\im}{\mathsf{i}}
\DeclareMathOperator{\supp}{supp}

\newcommand{\scp}[2]{\langle #1, #2 \rangle}

\newcommand{\ie}{\emph{i.e.}, }
\newcommand{\eg}{\emph{e.g.}, }

\usepackage{xcolor}
\usepackage[normalem]{ulem}

\usepackage[footnotesize]{caption}
\usepackage{graphicx,subcaption}

\usepackage{amsthm}

\newtheorem{theorem}{Theorem}

\newtheorem{lemma}[theorem]{Lemma}
\theoremstyle{definition}
\newtheorem*{definition}{Definition}

\usepackage{tikz}
\usetikzlibrary{calc}

\usepackage{graphicx}
\usepackage{amsmath}
\usepackage{xspace}

\newcommand{\whp}{w.h.p.\xspace}
\newcommand{\rv}{{\em r.v.}\xspace}
\newcommand{\rvs}{{\em r.v.}s\xspace}

\newcommand{\iid}{%
    \ifmmode
        \mathrm{i.i.d.}%
    \else%
        i.i.d.\@\xspace%
    \fi%
}
\DeclareMathOperator*{\dist}{\sim}
\newcommand{\al}[1]{\sign_{\bb C}(#1)}

\newcommand{\htex}{3in}
\newcommand{\sizep}{1pt}

\newcommand{\vs}{}

\begin{document}
\title{($\ell_1,\ell_2$)-RIP and Projected Back-Projection Reconstruction for Phase-Only Measurements}
\author{Thomas Feuillen, Mike E. Davies, Luc Vandendorpe, and Laurent Jacques
\thanks{TF, LV and LJ are with the ICTEAM institute,  UCLouvain,
Belgium (e-mail: \{thomas.feuillen, luc.vandendorpe, laurent.jacques\}@uclouvain.be).
MD is with the Institute for Digital Communications, University of Edinburgh, Edinburgh EH9 3JL, U.K. (e-mail: mike.davies@ed.ac.uk). MD would like to acknowledge partial support from the ERC project C-SENSE (ERC-ADG-2015-694888). LJ is funded by the F.R.S.-FNRS.}\vs \vs}
\maketitle
\begin{abstract}
This letter analyzes the performances of a simple reconstruction method, namely the Projected Back-Projection~(PBP), for estimating the direction of a sparse signal from its phase-only (or amplitude-less) complex Gaussian random measurements, \ie an extension of one-bit compressive sensing to the complex field. To study the performances of this algorithm, we show that complex Gaussian random matrices respect, with high probability, a variant of the Restricted Isometry Property~(RIP) relating to the $\ell_1$-norm of the sparse signal measurements to their $\ell_2$-norm. This property allows us to upper-bound the reconstruction error of PBP in the presence of phase noise. Monte Carlo simulations are performed to highlight the performance of our approach in this phase-only acquisition model when compared to error achieved by PBP in classical compressive sensing.
\end{abstract}
\IEEEpeerreviewmaketitle

\vs
\section{Introduction}
\label{sec:introduction}
One aspect of compressive sensing (CS) is to reduce the number of
measurements needed to achieve (high) quality reconstruction of
low-complexity signals (\eg sparse)~\cite{Don06,CRT06}. Recent
research has also focused on reducing the accuracy of each
measurement, \eg by lowering their resolution (or bit-depth) in specific
quantization contexts~\cite{BR08, GLP12, JLB13, JHF10}. This paper
investigates the consequences of removing the information about the
amplitude of a complex signal, \ie using only the measurement phase
for the reconstruction. While phase-only (PO) acquisition can serve
as a stepping stone to study new quantizations schemes, \eg when quantizing the
measurement phase~\cite{Bou13b}, this sensing is tantamount to a complex form of one-bit quantization, \eg extensively studied in one-bit CS~\cite{Fou16, JLB13, BR08}.

Oppenheim and co-authors~\cite{OL81, OHL82} proved in a few seminal contributions that real, bandlimited signals can be reconstructed, up to a lost amplitude, from the phase of their Fourier transform. More recently, for phase-only CS (PO-CS) with complex Gaussian random matrices, Boufounos determined that a specific distance between the measurement phases of two sparse signals encodes their angular distance up to an additive distortion~\cite{Bou13}. While this distortion prevents us from proving perfect estimation of sparse signal direction, the author showed experimentally that this achievable, thanks to a greedy algorithm enforcing the phase consistency between the signal estimate and the PO measurements. 

In this context, our contributions are as follows. While the question of perfect recovery of signal direction remains open, we here focus on a simple, non-iterative algorithm, the Projected Back-Projection~(PBP, see Sec.~\ref{sec:system-model}), and show that this method accurately estimates the direction of sparse signals in PO-CS (Sec.~\ref{sec:PBP}). This is possible if the sensing matrix respects a variant of the RIP, the ($\ell_1, \ell_2$)-RIP in the complex field, which was previously introduced for (real) one-bit CS. Using tools from measure concentration~\cite{LT91}, we then prove that complex Gaussian random matrices satisfy, with high probability (\whp), the ($\ell_1, \ell_2$)-RIP if the number of measurements is large compared to the signal sparsity level (Sec.~\ref{sec:ell_1-ell_2-rip}). Note that the $\ell_1$-norm of this RIP prevents a simple proof of this result by recasting the complex field to the real field. Finally, extensive Monte Carlo simulations confirm that the PBP estimation error for PO-CS compares favorably to the one of an unaltered, linear CS scheme~(Sec.~\ref{sec:simulations}).
\medskip

{\bf Notations and conventions:} We denote matrices and vectors
with bold symbols, \eg $\bs \Phi \in \mathbb{C}^{m\times n}$, $\bs x
\in \mathbb{C}^{n}$, and scalar values with light symbols. We will often use
the following quantities: $[d]:=\{1,\,\cdots, d\}$ with
$d \in \bb N$; the complex number $\im$ such that $\im^2=-1$;
$\Re\{\lambda\}$ (or $\lambda^\Re$) and $\Im\{\lambda\}$ (or $\lambda^\Im$) are the real and
imaginary part of $\lambda \in \bb C$, respectively, and
$\lambda^\ast$ is its complex conjugate; $\bs A^H$ is the
conjugate transpose of $\bs A$; $\supp \bs x$ is the support of $\bs x \in
\bb C^d$; $|\cl S|$ is the cardinality of a finite set $\cl S$;  $\langle\bs x, \bs y  \rangle= \sum_{i=1}^d x_i  y_i^\ast$ is the scalar product between two vectors $\bs x, \bs y \in \mathbb{C}^d$; the
$\ell_p$-norm of $\bs x$ ($p\geq 1$) is defined as $\|\bs x\|_p=(\sum_{i=1}^d
|x_i|^p)^{1/p}$, with $\|\bs x\|_\infty = \max_i |x_i|$ and $\|\bs u \|_0:= |\supp (\bs u) |$, and the $\ell_{p,q}$-norm of $\bs A = (\bs a_1, \cdots, \bs a_d)^\top \in \bb C^{d \times d'}$ is $\|\bs A\|_{p,q} = (\sum_{i=1}^d \|\bs
a_i\|_p^q)^{1/q}$; $\bar{\mathbb{B}}^n := \{\bs u \in \bb C^n: \|\bs u\| \leq 1\}$; the Hadamard product is $\odot$; and the angle operator (applied componentwise
onto vectors) reads $\angle(c e^{\im \phi}) = \phi$ for $c >0$ and $\phi \in [-\pi, \pi]$.
We denote by $\cl N^{m\times n}(\mu, \sigma^2)$ and $\bb C\cl N^{m\times n}(\mu, 2\sigma^2)$ (dropping the symbol $n$ if $n=1$) the $m\times n$ random matrices with entries independently and identically distributed ($\iid$) as the normal distribution $\cl N(\mu, \sigma^2)$ and the complex normal distribution $\bb C\cl N(\mu, 2\sigma^2) \sim N(\mu^\Re, \sigma^2) + \im N(\mu^\Im, \sigma^2)$, respectively, for some mean $\mu$ and variance $\sigma^2$. An $s$-sparse $\bs x$ vector belongs to the set $\bar\Sigma^n_s :=\{\bs u \in \mathbb{C}^n, \|\bs u \|_0 \leq s \}$. Given $g, g' \sim \mathcal{N}(0, \sigma^2)$, the random variable (\rv) $z:= |g+\im g'|$ is distributed as the Rayleigh distribution $\mathcal{R}(\sigma)$ with parameter~$\sigma$~\cite{PP02}.
\vs
\section{Phase-only sensing model}
\label{sec:system-model}
Let us consider a complex $s$-sparse vector $ \bs x_0 \in
\bar \Sigma^n_s$. Given a complex matrix $\bs \Phi \in \bb C^{m \times
n}$, this work is concerned with the following noisy
non-linear sensing model~\cite{Bou13}, which generalizes one-bit CS
\cite{Fou16, XJ18} to the complex field:
\begin{equation}
\label{eq:phase-only-sensing-model}
\bs z = \al{\bs \Phi \bs x_0} \odot  e^{\im \bs \xi},
\end{equation}
where $\al{\cdot}$ is the \emph{complex signum} operator, applied
component-wise onto vectors, \ie $\al{\lambda} = \lambda/|\lambda|$ for
$\lambda \in \bb C \setminus \{0\}$, and $\bs \xi$ stands for a possible corruption
of the measurement phase (with $\xi_i \in [0, 2\pi)$, $i\in [m]$). 
The matrix $\bs \Phi$ can be, \eg a complex Gaussian
random matrix (see Sec.~\ref{sec:ell_1-ell_2-rip}). 

The sensing model~\eqref{eq:phase-only-sensing-model} thus
discards the amplitudes of the measurements $\bs \Phi
\bs x_0$; estimating $\bs x_0$ from $\bs z$ is possible
only up to a global unknown normalization of $\bs x_0$, \ie only the direction $\bs x_0/\|\bs x_0\|_2$ can be estimated. 

We aim to show that the projected back projection (PBP) algorithm
\cite{XJ18,Fou16} accurately estimates the direction of complex sparse signals
provided the complex sensing matrix respects a variant of the RIP property (see
Sec.~\ref{sec:PBP}). Given $s \in [n]$, the sensing matrix $\bs \Phi$, and the
measurement vector $\bs z$, this algorithm is simply defined as
\begin{equation}
\label{eq:PBP}
\hat{\bs x}= {\sf H}_s \big( \bs \Phi^H \bs z \big),
\tag{\rm PBP}
\end{equation}
where ${\sf H}_{s}(\bs u)$ is the hard thresholding operator setting all
of the components of the vector $\bs u$ to zero but the $s$ strongest in
amplitude (which are unchanged). For CS,~\ref{eq:PBP} is often used as the first iteration of more complex iterative
methods such as iterative hard thresholding (IHT)
\cite{BD09,XJ18}. Despite its simplicity, analyzing~\ref{eq:PBP}
can thus lead to better iterative reconstruction algorithms for PO-CS.
\vs
\section{Bound on the PBP reconstruction error}
\label{sec:PBP}
In CS theory, the error of most signal reconstruction algorithms is
controlled by the restricted isometry property --- or ($\ell_2,
\ell_2$)-RIP --- of the sensing matrix~\cite{FH13}. This amounts to
asking that for some $\delta > 0$, 
$$
(1-\delta)\| \bs x \|_2^2\leq
\|\bs \Phi \bs x\|_2^2 \leq (1+\delta)\| \bs x \|_2^2,
$$ 
holds true for all sparse vectors~$\bs x$. For instance, if the (real or complex) matrix $\bs \Phi$ respects the ($\ell_2,\ell_2$)-RIP over all $2s$-sparse vectors and one observes a $s$-sparse vector from the model $\bs y = \bs \Phi \bs x_0$, the error of the estimate $\hat{\bs x} = H_s(\bs \Phi^H \bs y)$ is bounded as $\|\bs x_0 - \hat{\bs x}\| = O(\delta)$~\cite{FH13,XJ18}.  

As will be clear below, the capacity of~\ref{eq:PBP} to
estimate a sparse vector $\bs x_0$ from its complex, phase-only observations $\bs z$ in~\eqref{eq:phase-only-sensing-model} depends on the following RIP variant.
\begin{definition}
\label{def:ripl1l2}
Given $\delta>0$, the matrix $\bs \Phi \in \mathbb{C}^{m \times n}$ satisfies the ($\ell_1,\ell_2$)-RIP($s$, $\delta$) if, for all $\bs x \in \bar\Sigma^n_s$, 
\begin{equation}
(1-\delta)\|\bs x \|_2 \leq \| \bs \Phi \bs x\|_1 \leq (1+\delta)\|\bs x \|_2.
\end{equation}
\end{definition}
This property was introduced for real one-bit CS~\cite{Fou16,Pla14}; with it, specific algorithms
(including~\ref{eq:PBP}) yield a good estimate of a real sparse signal from the sign of its
random measurements. Moreover, provided that $m$ is large compared to $s$, different types of real random matrix
constructions, such as Gaussian random matrices~\cite[Lemma
2.1]{Pla14}\cite{JHF10} or randomly subsampled Gaussian circulant
matrices~\cite{Dir17}, have been shown to respect the
($\ell_1,\ell_2$)-RIP($s$, $\delta$) \whp. 

To bound the reconstruction error of~\ref{eq:PBP}, we first need the following
lemma that is adapted from~\cite[Lemma 3]{Fou16}.
\begin{lemma}
\label{lem:implication-of-RIPL1L2} 
If $\bs \Phi$ satisfies the ($\ell_1,\ell_2$)-RIP$(\delta,s)$ for
$0<\delta<1$ and $s\in[n]$, then for any vector
$\bs x \in \mathbb{C}^n$ with unit $\ell_2$-norm such that $\supp \bs x \subset \mathcal S \subset
[n]$ with $|\cl S| = s$, 
\begin{equation}
\label{eq:pbp-oracle-error}
\ts \big\|{\sf H}_{\cl S}(\bs \Phi^H \al{\bs \Phi \bs x}  \big)  - \bs x \big\|_2 \leq \sqrt{5 \delta}.
\end{equation}
\end{lemma}

We can now determine the main result of this section, which derives
from an adaptation of~\cite[Thm 8]{Fou16} to the complex field.  
\begin{theorem}
  \label{thm:bound-pbp-reconstr}
  If $\bs \Phi$ satisfies ($\ell_1,\ell_2$)-RIP($2s$,$\delta$), then the~\ref{eq:PBP} estimate
  $\hat{\bs x}$ of any signal $\bs x_0 \in \bar\Sigma^n_s$ with $\|\bs x_0\|_2 = 1$ observed via~\eqref{eq:phase-only-sensing-model} with $\|\bs \xi\|_\infty
  \leq \tau$ respects 
\begin{equation}
\label{eq:nooracle}
\|\bs x_0 -\hat{\bs x} \|_2 \leq 2\sqrt{5 \delta} + 4\tau.
\end{equation}
\end{theorem}
\begin{proof}
Let $\cl S_0$ and $\cl T$ be the $s$-sparse supports of $\bs x_0$ and $\hat{\bs x}$,
respectively. Writing $\cl S := \cl S_0 \cup \cl T$ (with $|\cl S|
\leq 2s$) and $\bs a = \bs \Phi^H \bs z$, we first note that
$\|\bs x_0 - \hat{\bs x}\|_2 \leq \|\bs x_0 - {\sf H}_{\cl S}(\bs a)\|_2 +
\|\hat{\bs x} - {\sf H}_{\cl S}(\bs a)\|_2$, so that $\|\bs x_0 - \hat{\bs x}\|_2 
\leq 2\|\bs x_0 - {\sf H}_{\cl S}(\bs a)\|_2$ since $\hat{\bs x}$ is the best $s$-term approximation of
both $\bs a$ and ${\sf H}_{\cl S}(\bs a)$. The triangular inequality and Lemma~\ref{lem:implication-of-RIPL1L2} then provide
\begin{multline}
  \ts \|\bs x_0 - {\sf H}_{\cl S}(\bs a)\|_2 = \|\bs x_0 -{\sf H}_{\cl S}\big(\bs \Phi^H [\al{\bs \Phi \bs x_0} \odot  \exp(\im \bs \xi)] \big)\|_2 \\
\leq \sqrt{5 \delta} + \|{\sf H}_{\cl S}\big(\bs \Phi^H [\al{\bs \Phi \bs x_0} \odot (1-  e^{\im \bs \xi})]\big)\|_2.
\end{multline}
Since $\bs \Phi$ respects the ($\ell_1,\ell_2$)-RIP($2s$,$\delta$), we get
\begin{align}
&\ts\|{\sf H}_{\cl S}\big(\bs \Phi^H [\al{\bs \Phi \bs x_0} \odot (1-
e^{\im \bs \xi})]\big)\|_2\\
&\quad\ts= \sup_{\bs u \in \bar{\bb B}^n} \scp{ \bs \Phi ({\sf H}_{\cl S}(\bs
  u))}{ \al{\bs \Phi \bs x_0} \odot (1-
  e^{\im \bs \xi})}\\
&\quad\ts\leq \|1-  e^{\im \bs \xi}\|_\infty\ \sup_{\bs u \in \bar{\bb B}^n}\ \|\bs \Phi ({\sf H}_{\cl S}(\bs
u))\|_1\\
&\quad\ts\leq 2 \|1-  e^{\im \bs \xi}\|_\infty \leq 2 \|\bs \xi\|_\infty \leq 2 \tau.
\end{align}
Gathering all bounds provides the result.
\end{proof}
Interestingly,~\eqref{eq:nooracle} shows that one can still accurately estimate the direction of a complex sparse signal in PO-CS if $\bs \Phi$ is ($\ell_1,\ell_2$)-RIP$(2s, \delta)$ with a small constant~$\delta$. 

Moreover, as clarified in Sec.~\ref{sec:ell_1-ell_2-rip}, \eqref{eq:nooracle} allows us to understand how, for complex Gaussian sensing matrices, the error of PBP decays when $m$ increases. Indeed, up to some missing log factors, we prove in Thm.~\ref{thm:ripl1l2-for-gaussian} that complex Gaussian random matrices satisfy the ($\ell_1,\ell_2$)-RIP$(2s, \delta)$ \whp provided $m \geq C \delta^{-2} s$ for some $C>0$. By saturating this condition, we see that, for noiseless PO-CS,~\ref{eq:PBP} achieves the error
\begin{equation}
\label{eq:PBP-decay-in-m}
\|\bs x_0 - \hat{\bs x}\|_2 = O\big(\sqrt[4]{s/m}\big)    
\end{equation}
when $m$ increases, \ie which tends to zero for large $m$. 

This evolution of the PBP error meets the one encountered for real one-bit CS~\cite{Fou16} and non-linear CS~\cite{PV16}. However, this behavior is a bit pessimistic compared to the experimental decay in $O(\sqrt{s/m})$ reached by simulations (see Sec.~\ref{sec:simulations}). The exponent over $\delta$ in~\eqref{eq:nooracle} could thus be improved from $\sqrt \delta$ to~$\delta$. This would then match the performances of PBP in linear CS (see the beginning of this section) and dithered quantized CS~\cite{XJ18,JC17} where it reaches an error bounded by $O(\delta)$ for ($\ell_2,\ell_2$)-RIP$(2s, \delta)$ sensing matrices, \ie a decay in $O(\sqrt{s/m})$ for Gaussian random sensing matrices. 
%
\vs
\section{The ($\ell_1,\ell_2$)-RIP of Complex Gaussian Matrices}
\label{sec:ell_1-ell_2-rip}
While one easily extends the ($\ell_2, \ell_2$)-RIP of certain random matrix constructions from the real to the complex fields --- \eg by recasting the signal space $\mathbb{C}^n$ and measurement domain $\mathbb{C}^m$ to $\mathbb{R}^{2n}$  and $\mathbb{R}^{2m}$, respectively~\cite{FH13} --- such an extension for ($\ell_1, \ell_2$)-RIP matrices is not known. 


Fortunately, using the tools of measure concentration~\cite{LT91}, we prove below that complex Gaussian random matrices $\bs \Phi$ respects the ($\ell_1,\ell_2$)-RIP \whp provided $m$ is large compared to the signal sparsity. To show this, we first establish that, given $\bs x \in \bb C^n$, $\bb E\|\bs \Phi \bs x \|_1$ is proportional to $\|\bs x\|_2$ since each random variable $|(\bs \Phi \bs x)_i|$ is Rayleigh distributed.
\begin{lemma}
\label{lem:Rayleigh}
Given $\bs x \in \mathbb{C}^n$  and a random matrix $\bs \Phi \sim
\mathbb{C}\mathcal{N}^{m\times n}(0,\sigma^2)$ with $\sigma :=\frac{1}{m} \frac{\sqrt{2}}{\sqrt{\pi}}$, we have 
\begin{equation}
\mathbb{E}\big[ \| \bs \Phi \bs x \|_1 \big]= \| \bs x \|_2.
\label{eq:expectation_l1}
\end{equation}
\end{lemma}
\vs
\begin{proof}
By decomposing both the entries of $\bs \Phi$ and the components of $\bs x $ into their real and imaginary parts, we get
$$
\ts \| \bs \Phi \bs x \|_1  = \sum_{i=1}^m |\sum_{j=1}^n  \Phi_{ij} x_j | =\sum_{i=1}^m \big |\sum_{j=1}^n g^\Re_{ij} + \im g^\Im_{ij} \big|,
$$
with $g^\Re_{ij}:= \Phi^\Re_{ij}  x^\Re_j - \Phi^\Im_{ij}  x^\Im_j$ and $g^\Im_{ij} := \Phi^\Re_{ij}  x^\Im_j + \Phi^\Im_{ij}  x^\Re_j$.

We note that, for all indices $i,i' \in [m]$ and $j,j' \in [n]$,  $\Phi_{ij}^\Re$ and $\Phi_{ij}^\Im$ are Gaussian random variables with $\mathbb{E}[\Phi_{ij}^\Re ]=\mathbb{E}[\Phi_{ij}^\Im ]=\mathbb{E}[\Phi_{ij}^\Re \Phi_{i'j'}^\Im ]= 0$. 
Therefore, $g^\Re_{ij}, g^\Im_{ij} \dist_{\iid}
\mathcal{N}(0,\sigma^2 |x_j|^2)$ and  a simple computation provides $\bb E g^\Re_{ij}
g^\Im_{i'j'} = 0$. The \rvs 
$\Gamma^\Re_i:=\sum_{j=1}^n g^\Re_{ij}  $ and
$\Gamma^\Im_i:=\sum_{j=1}^n g^\Im_{ij}  $  are
thus independent and distributed as $\mathcal{N}(0,\sigma^2
\|\bs x \|_2^2)$ for all $i\in[m]$. Consequently, 
$$
\ts \mathbb{E}\big[ \| \bs \Phi \bs x \|_1 \big] = \sum_{i=1}^m \mathbb{E}\big[| \Gamma^\Re_i+\im \Gamma^\Im_i  | \big]= m  \mathbb{E}\big[ \Gamma_0 \big],
$$
where $\Gamma_0$ follows a Rayleigh distribution $\mathcal{R}(\sigma
\|\bs x \|_2)$. Since $\mathbb{E}[\Gamma_0]= \sigma
\sqrt{\frac{\pi}{2}} \|\bs x \|_2 $~\cite{PP02} and
$\sigma=\frac{1}{m}\sqrt{\frac{2}{\pi}}$, we find $\mathbb{E}\big[ \| \bs \Phi \bs x \|_1 \big]= \sigma \| \bs x \|_2 \sqrt{\frac{\pi}{2}}m=\|\bs x \|_2$.\vspace{2mm}
\end{proof}

We also need this classical result from Ledoux and Talagrand
\cite[Eq. 1.6]{LT91}, see also~\cite[Lemma 5]{JHF10}. 
\begin{lemma}
\label{lem:LD-concent} If the function $F$ is Lipschitz with $\lambda=\| F \|_{\text{Lip}}$, then, for $r>0$ and $\bs \gamma \sim \mathcal{N}^{m}(0,1)$, 
\begin{equation}
\ts \bb P \big( \big| F(\bs \gamma) -\bb E(F(\bs \gamma)) \big| > r\big) \leq 2 \exp(-\frac{1}{2} r^2 \lambda^{-2}).
\end{equation}
\end{lemma}
In our developments, $F$ will be of the following kind. 
\begin{lemma}
\label{lem:lipschitz}
 The functions $G: \bs u \in \mathbb{C}^m  \mapsto \|\bs u \|_1 \in \mathbb{R}_{+}$ and of $G' : (\bs u^\Re,\bs u^\Im) \in \mathbb{R}^{m \times 2} \mapsto \|(\bs u^\Re,\bs u^\Im)\|_{2,1}  \in \mathbb{R}_{+} $ have a Lipschitz constant equal to $\sqrt{m}$.
\end{lemma}
\begin{proof}
For all $\bs u, \bs v \in \mathbb{C}^m$, $|\|\bs u \|_1-\|\bs v \|_1|\leq \|\bs u -\bs v
\|_1 \leq \sqrt{m}\|\bs u -\bs v \|_2 $, which gives the Lipschitz
constant of $G$. The one of $G'$ follows from $\|\bs u  \|_1=\| (\bs u^\Re,\bs u^\Im)\|_{2,1}$ .
\end{proof}
\medskip
We are now ready to prove the main result of this section.
\begin{theorem}
\label{thm:ripl1l2-for-gaussian}
Let $\delta \in (0,1)$, $\sigma = \frac{1}{m} \frac{\sqrt{ 2}}{\sqrt{\pi}}$, and $\bs \Phi \sim
  \mathbb{C}\mathcal{N}^{m\times n}(0, \sigma^2)$  be a complex
  Gaussian random matrix.
  If $m\geq \frac{36}{\pi} \delta^{-2} \big[s \log\big(\frac{en}{s} (1+\frac{6}{\delta})^2 \big)+\log(\frac{2}{\eta}) \big] $, then, with probability exceeding $1-\eta$, the matrix $\bs \Phi$ satisfies the ($\ell_1, \ell_2$)-RIP($s, \delta$).
\end{theorem}
\begin{proof}
The proof strategy follows the one developed in~\cite{BDDW08} for
proving that real Gaussian random matrices satisfy the
($\ell_2,\ell_2$)-RIP \whp. By homogeneity of the $(\ell_1,\ell_2)$-RIP, it is enough to
prove that complex Gaussian random matrices satisfy it \whp for all vectors of $\tilde\Sigma^{n}_{s} :=
\bar\Sigma^n_s \cap \bar{\mathbb{B}}^n$.

We first show that for a fixed vector
$\bs x \in \mathbb{C}^n$, $\|\bs \Phi \bs x \|_1$ concentrates around
$\| \bs x \|_2$. Using the \rvs $\Gamma^\Re_i, \Gamma^\Im_i$ defined in the proof of Lemma~\ref{lem:Rayleigh}, we can write 
\begin{align}
p&\ts:=\mathbb{P}\big(\,\big| \|\bs \Phi \bs x \|_1 - \| \bs x \|_2 \big| > t \| \bs x \|_2 \big)\\
  &\ts=\mathbb{P}\big(\, \big| \sum_{i=1}^m  \big((\Gamma^\Re_i)^2 +  (\Gamma^\Im_i)^2\big)^{1/2}- \| \bs x \|_2 \big| > t \| \bs x \|_2 \big)\\
  &\ts=\mathbb{P}\big(\, \big| \sum_{i=1}^m \big((\gamma^\Re_i)^2 +  (\gamma^\Im_i)^2\big)^{1/2} - m \sqrt{\frac{\pi}{2}} \big| > t m  \sqrt{\frac{\pi}{2}} \big),
\end{align}
where we defined the independent Gaussian random vectors $\bs
\gamma^\Re,\bs \gamma^\Im\dist_{\iid} \mathcal{N}^{m }(0,1)$. Since $\sum_{i=1}^m \big((\gamma^\Re_i)^2 +  (\gamma^\Im_i)^2\big)^{1/2} = \|(\bs \gamma^\Re,\bs \gamma^\Im) \|_{2,1}$, Lemma~\ref{lem:LD-concent} provides
\begin{align}
p&\ts = \mathbb{P}\big(\,\big|\,\|(\bs \gamma^\Re,\bs \gamma^\Im) \|_{2,1} - m \sqrt{\frac{\pi}{2}} \big| > t m  \sqrt{\frac{\pi}{2}} \big)\\
&\ts \leq 2\exp{\big(-\frac{\pi}{4} t^2 m\big)}
\label{eq:concentr}
\end{align}
by considering $\bs \gamma = (\bs \gamma^\Re,\bs \gamma^\Im)$ as a
$2m$ Gaussian random vector, with the function $F(\bs \gamma) :=\| (\bs
\gamma^\Re,\bs \gamma^\Im)\|_{2,1}$ whose Lipschitz constant is
characterized in Lemma~\ref{lem:lipschitz}. Therefore, given $\bs x$
and $t > 0$,
we have
$$
\big| \|\bs \Phi \bs x \|_1 - \| \bs x \|_2 \big| \leq t \| \bs x \|_2
$$
with probability exceeding $1-p \geq 1 - 2\exp{\big(-\frac{\pi}{4} t^2 m\big)}$.

We now extend this result to all vectors of $\tilde\Sigma^{n}_{s}$ 
by first determining when this concentration holds for all the vectors of a
$\rho$-covering of this domain --- that is a set such that all elements
of $\tilde\Sigma^{n}_{s}$ are no more than $\rho > 0$ far apart from an
element of this covering --- and by finally extending this property to $\tilde\Sigma^{n}_{s}$ by continuity.

Let us first build this covering. We note that $\tilde\Sigma_s^{n}=
\bigcup_{\cl S \subset [n]: |\cl S| = s}
\tilde\Sigma^{n}(\mathcal{S})$, with
$\tilde\Sigma^{n}(\mathcal{S}) := \{\bs u \in \bar{\bb B}^n: \supp
\bs u = \cl S\}$. Moreover, $\tilde\Sigma^{n}(\mathcal{S})$ is
isomorphic to $\bar{\mathbb{B}}^s$, and thus to
$\mathbb{B}^{2s}$. Since this last set, and thus $\tilde\Sigma^{n}(\mathcal{S})$, can be covered with no more than
$(1+\frac{2}{\rho})^{2s}$ vectors~\cite{BDDW08}, a covering $\cl J_\rho$ of
$\tilde\Sigma_s^{n}$ can be reached by gathering all coverings --- ${n
  \choose s}$ in total --- so
that
$$
\ts |\cl J_\rho| \leq {n \choose s} (1+\frac{2}{\rho})^{2s} \leq
(\frac{en}{s})^s (1+\frac{2}{\rho})^{2s}.
$$
Interestingly, by design, this covering is such that all
$\bs x \in \tilde\Sigma_s^{n}$ can be written as $\bs x = \bs u +
\bs r$ with $\bs u \in \cl J_\rho \subset \tilde\Sigma_s^{n}$, $\bs
r \in \rho \bar{\bb B}^n\cap \tilde\Sigma_s^{n} = \rho \tilde\Sigma_s^{n}$, with $\supp \bs x =
\supp \bs u = \supp \bs r$.

Using~\eqref{eq:concentr}, by union bound over all the vectors of $\cl J_\rho$, the
event 
\begin{equation}
  \label{eq:event-concent-cover}
  \cl E_{\rho,t}:\quad \big | \|\bs \Phi \bs u \|_1  - \| \bs u \|_2 |\ \leq\ t,\quad \forall \bs u \in \cl J_\rho,  
\end{equation}
holds with failure probability $p_{\rho,t} := \bb P(\cl E_{\rho,t}^c)$ at most
\begin{equation}
\label{eq:concentr_cover}
\ts p_{\rho,t} \leq 2\big(\frac{en}{s}\big)^s\big(1+\frac{2}{\rho}\big)^{2s}  \exp{\big(-\frac{\pi}{4} t^2 m\big)}.
\end{equation}

Let us assume $\cl E_{\rho,t}$ holds and pick an arbitrary $\bs x \in
\tilde\Sigma^{n}_s$. As explained above, we can write $\bs x = \bs u + \bs r$ with $\bs u \in \cl J_\rho$, $\bs
r \in \rho \tilde\Sigma_s^{n}$, and $\supp \bs x =
\supp \bs u = \supp \bs r$. 

Using~\eqref{eq:event-concent-cover}, and the properties of the covering, we get
 \begin{align}
&\ts | \|\bs \Phi \bs x \|_1 -\|\bs x \|_2|= | \|\bs \Phi (\bs u +\bs r)  \|_1 -\|(\bs u +\bs r) \|_2| \\
&\ts\leq|\|\bs \Phi \bs u\|_1-\|\bs u\|_2|+ |\|\bs \Phi (\bs u +\bs r)\|_1  -\|\bs \Phi \bs u\|_1|\\
&\ts+| \|\bs u + \bs r\|_2- \|\bs u\|_2| \ts \leq t +\rho + \rho \|\bs \Phi (\rho^{-1}\bs r)\|_1,
\label{eq:trickA}
 \end{align}
where we used multiple times the triangular inequality. However, $\rho^{-1}\bs r \in \tilde\Sigma_s^{n}$ and we can recursively apply the same development to $\|\bs \Phi (\rho^{-1}\bs r)\|_1$, so that     
$$
\ts | \|\bs \Phi \bs x \|_1 -\|\bs x \|_2| \leq (t+\rho) \sum_{k=0}^{+\infty} \rho^k = \frac{t+\rho}{1 - \rho}.
$$

Setting $t=\rho=\delta/3$ for some $0<\delta<1$, we get $\frac{t+\rho}{1 - \rho} \leq \delta$. From the analysis of $\cl E_{\rho,t}$ above, we finally obtain that $| \|\bs \Phi \bs x \|_1 -\|\bs x \|_2| \leq \delta$ holds true for all $\bs x \in \tilde \Sigma^n_s$ --- \ie the ($\ell_1,\ell_2$)-RIP is verified --- with failure probability at most 
$$
\ts p_{\frac{\delta}{3}, \frac{\delta}{3}} \leq   2\big(\frac{en}{s}\big)^s\big(1+\frac{6}{\delta}\big)^{2s}  \exp{\big(-\frac{\pi}{36} \delta^2 m\big)}.
$$
We conclude the proof by observing that $p_{\frac{\delta}{3}, \frac{\delta}{3}} \leq \eta$ for $0<\eta<1$ if $m\geq \frac{36}{\pi} \delta^{-2} \big[s \log\big(\frac{en}{s} (1+\frac{6}{\delta})^2 \big)+\log(\frac{2}{\eta}) \big]$
\end{proof}
\vs
\section{Simulations}
\label{sec:simulations}
We now assess the tightness of our theoretical analysis through Monte Carlo simulations. We do not aim to demonstrate the superiority of~\eqref{eq:PBP} over other methods but to study the potentialities of such a simple algorithm in PO-CS.

As a first experiment, we have tested the estimation of complex sparse signals $\bs x_0$ in $\bb C^n$ with $n=256$ for different sparsity levels $s \in [n]$ and measurement number $m$. Two acquisition strategies were compared: the phase-only acquisition fixed by the model~\eqref{eq:phase-only-sensing-model}, and classical compressive sensing where we directly acquire the measurement vector $\bs y := \bs \Phi \bs x_0$ without alteration.  For each combination of $s$ and $m$, the performances of both strategies have been tested over 100\,000 generations of the sparse signal $\bs x_0$ and the complex Gaussian random matrix $\bs \Phi \sim \bb C\cl N(0, \sigma^2)$, with $\sigma^2$ set to $2/(\pi m^2)$ and $1/m$ for the phase-only and the CS scenario, respectively. Each sparse signal $\bs x_0$ was created by picking a $s$-sparse support uniformly at random amongst the $n \choose s$ possible supports, inserting in this support $s$ \iid complex values picked uniformly at random 
before normalizing.
We analyzed the reconstruction error of the signal direction with the metric $\cl E(\bs x_0, \hat{\bs x}) := \|\bs x_0 -\|\hat{\bs x}\|_2^{-1} \hat{\bs x}\|_2$, where $\hat{\bs x}$ is the~\eqref{eq:PBP} estimate.
\begin{figure}[!t]
\centering
\begin{tikzpicture}

\definecolor{color0}{rgb}{1,0.647058823529412,0}
\definecolor{color1}{rgb}{0.933333333333333,0.509803921568627,0.933333333333333}

\begin{axis}[
width=0.6\columnwidth,
height=0.8*\htex,
tick align=outside,
tick pos=left,
x grid style={white!69.01960784313725!black},
xlabel={\scriptsize $\log_2(\frac{m}{n})$},
xmin=-6., xmax=4.,
xtick style={color=black},
y grid style={white!69.01960784313725!black},
ylabel={\scriptsize $10\log_{10}(\|\bs x_0 -\frac{\hat{\bs x}}{\|\hat{\bs x}\|}\|_2)$},
ymin=-19, ymax=11.,
ytick style={color=black}
]
\addplot [thick, red, dashed,mark=*, mark size=\sizep, mark options={solid}]
table {%
-6 1.11243847158974
-4 -2.19602027556124
-2 -7.21078128648834
0 -11.2878435263983
2 -14.8777905606686
4 -18.2258465459564
};
\addplot [thick, red,mark=*, mark size=\sizep, mark options={solid}]
table {%
-6 1.26096659043204
-4 -1.46289792447115
-2 -6.44871320259678
0 -10.6256988444573
2 -14.28000096159
4 -17.6637072106488
};
\addplot [thick, blue, dashed,mark=*, mark size=\sizep, mark options={solid}]
table {%
-6 1.25882372315679
-4 -0.42435704946086
-2 -4.10268012690458
0 -8.11111528471988
2 -11.8600898785594
4 -15.3232364836426
};
\addplot [thick, blue,mark=*, mark size=\sizep, mark options={solid}]
table {%
-6 1.3298168828923
-4 0.00912456760578186
-2 -3.41747972403624
0 -7.43630494697117
2 -11.2309103715808
4 -14.739141598999
};
\addplot [thick, green!50.19607843137255!black, dashed,mark=*, mark size=\sizep, mark options={solid}]
table {%
-6 1.33174504817244
-4 0.59960328327647
-2 -1.86019579390291
0 -5.35478423913879
2 -9.00375098167221
4 -12.5449196121235
};
\addplot [thick, green!50.19607843137255!black,mark=*, mark size=\sizep, mark options={solid}]
table {%
-6 1.36232950715413
-4 0.812953461491129
-2 -1.32415033756304
0 -4.72938729154686
2 -8.37299257975443
4 -11.9387273204955
};
\addplot [thick, color0, dashed,mark=*, mark size=\sizep, mark options={solid}]
table {%
-6 1.33367699836556
-4 0.906583712893597
-2 -0.736439599950016
0 -3.79609287988756
2 -7.25990834708337
4 -10.7413476672853
};
\addplot [thick, color0,mark=*, mark size=\sizep, mark options={solid}]
table {%
-6 1.35545835788413
-4 1.02524291184372
-2 -0.327397312746119
0 -3.21918462005425
2 -6.64941621691661
4 -10.1414904429507
};
\addplot [thick, color1, dashed,mark=*, mark size=\sizep, mark options={solid}]
table {%
-6 1.30116259967353
-4 1.00724012839566
-2 0.0880021604146752
0 -2.19037987382712
2 -5.31269074598933
4 -8.65075941206365
};
\addplot [thick, color1,mark=*, mark size=\sizep, mark options={solid}]
table {%
-6 1.32090297666823
-4 1.07830127501019
-2 0.32766715710494
0 -1.70674824435146
2 -4.74255942880031
4 -8.06840631526832
};
\addplot [thick, white!50.19607843137255!black, dash pattern=on 1pt off 3pt on 3pt off 3pt]
table {%
-6 8.98970004336019
-4 5.97940008672038
-2 2.96910013008056
0 -0.0411998265592484
2 -3.05149978319906
4 -6.06179973983887
};
\addplot [thick, black, dash pattern=on 1pt off 3pt on 3pt off 3pt]
table {%
-6 10.4948500216801
-4 8.98970004336019
-2 7.48455006504028
0 5.97940008672038
2 4.47425010840047
4 2.96910013008056
};
\end{axis}

\end{tikzpicture}
\caption{(Best viewed in color) Reconstruction error of~\eqref{eq:PBP} for different measurement models. (dashed lines) compressive sensing; (solid lines) phase-only measurements. The colors represent the sparsity, namely $s=2$ in red, $s=4$ in blue, $s=10$ in green, $s=20$ in orange, and $s=50$ in pink. The dotted lines represent the rates of $m^{-\frac{1}{2}}$ in gray and $m^{-\frac{1}{4}}$ in black.\vs \vs
\label{fig:rec-error-pbp-CS-vs-po}}
\end{figure}
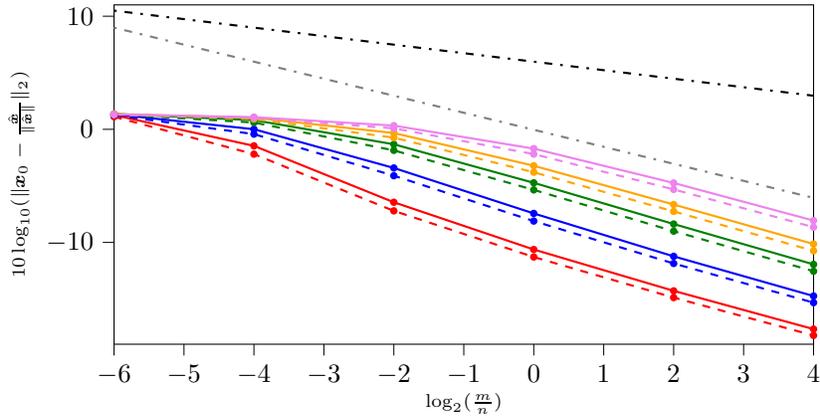
Comparing the two schemes in Fig.~\ref{fig:rec-error-pbp-CS-vs-po} for different sparsity levels, we observe that the reconstruction error achieved from phase-only measurements exhibits good performances given the absence of the amplitude information. The experimental convergence rate is also matching the one of the CS scheme; it scales as $m^{-\frac{1}{2}}$ when $m$ increases instead of the pessimistic rate in $m^{-\frac{1}{4}}$ predicted by the theory in~\eqref{eq:PBP-decay-in-m}. The phase-only scheme seems to only suffer from a constant loss (in dB) when compared to the classic model.    
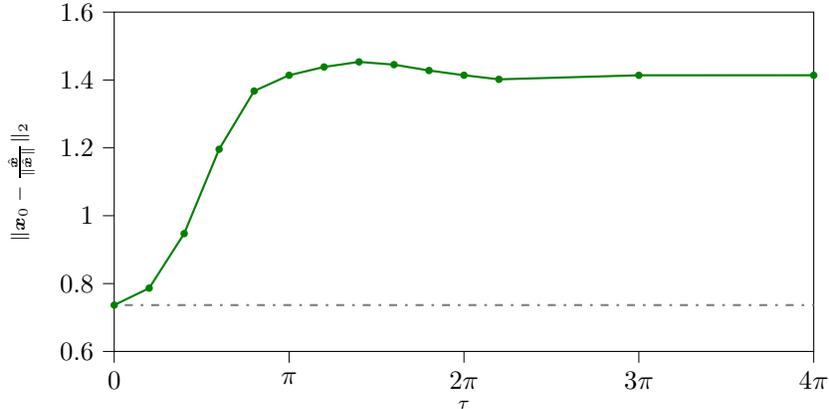
\begin{figure}[!t]
    \centering \centering
\begin{tikzpicture}

\begin{axis}[
width=0.6\columnwidth,
height=0.8*\htex,
tick align=outside,
tick pos=left,
x grid style={white!69.01960784313725!black},
xlabel={\scriptsize $\tau$},
xmin=0, xmax=4,
xtick style={color=black},
y grid style={white!69.01960784313725!black},
ylabel={\scriptsize $\|\bs x_0 -\frac{\hat{\bs x}}{\|\hat{\bs x}\|}\|_2$},
ymin=0.6,
ymax=1.6,
ytick style={color=black},
xtick={0,1,2,3,4},
xticklabels={$0$,$\pi$,$2\pi$,$3\pi$,$4\pi$}
]
\addplot [thick, green!50.19607843137255!black,mark=*, mark size=\sizep, mark options={solid}]
table {%
0 0.736640888843304
0.2 0.786610771897557
0.4 0.947341417436198
0.6 1.19585077913719
0.8 1.36743298280552
1 1.41399462231784
1.2 1.4384026758848
1.4 1.45333136330619
1.6 1.44523528092943
1.8 1.42798469530698
2 1.41398840688393
2.2 1.40190177210711
3 1.41389069329749
4 1.41397628693168
};
\addplot [thick, white!50.19607843137255!black, dash pattern=on 1pt off 3pt on 3pt off 3pt]
table {%
0 0.736640888843304
0.2 0.736640888843304
0.4 0.736640888843304
0.6 0.736640888843304
0.8 0.736640888843304
1 0.736640888843304
1.2 0.736640888843304
1.4 0.736640888843304
1.6 0.736640888843304
1.8 0.736640888843304
2 0.736640888843304
2.2 0.736640888843304
3 0.736640888843304
4 0.736640888843304
};
\end{axis}

\end{tikzpicture}
    \caption{Reconstruction error of~\eqref{eq:PBP} for noiseless (dashed lines) and noisy measurements (solid lines) for different $\tau$ with $s=10$ and $M=64$.\vs \vs}
    \label{estnoise}
\end{figure}

In a second experiment, we have studied the performances of PBP in the presence of phase noise. In this new test, we kept the same parameters as above, restricting only the sparsity level and the number of measurements to $s=10$ and $m=64$, respectively. The phase noise $\bs \xi$ in~\eqref{eq:phase-only-sensing-model} was generated according to a uniform distribution between $-\tau$ and $\tau$, with $\tau \in [0, 4\pi]$. As established~\eqref{thm:bound-pbp-reconstr}, the reconstruction error $\cl E(\bs x_0, \hat{\bs x})$ increases almost linearly when $\tau$ increases from 0 to $\pi$, before saturating at $\sqrt 2$ from $\tau > \pi$. In other words, from that noise level, phase-only measurements are too noisy and $\scp{\bs x_0}{\hat{\bs x}} \approx 0$. 
Furthermore, the additive nature of the degradation in~\eqref{thm:bound-pbp-reconstr} is clearly visible when comparing the noiseless in dashed gray and noisy reconstruction in solid green.
\vs
\section{Conclusion}
\label{sec:conclusion}
In this paper, we have studied how to estimate the direction of complex sparse vectors from noisy phase-only measurements. We proved theoretically that the estimate yielded by the projected back projection of noisy phase-only measurement has bounded and stable reconstruction error provided that the sensing matrix satisfies an extension of the ($\ell_1, \ell_2$)-RIP in the complex field. Moreover, we showed that $m \times n$ complex Gaussian random matrices respect \whp this property with distortion $\delta > 0$ provided that $m$ is large compared to the signal sparsity level $s$, \ie $m = O(\delta^{-2} s \log(\frac{n}{\delta s}))$. The proof of this result leverages the tools of measure concentration since the $\ell_1$-norm prevents a simple recasting of the complex ($\ell_1, \ell_2$)-RIP to a real domain of larger dimension. We finally analyzed the tightness of our theoretical developments through Monte Carlo simulations. They confirmed that, despite the lack of amplitude information, we can reach arbitrary high accuracy on the estimation of sparse signal direction provided $m/s$ is large, with an experimental error rate decaying as $1/\sqrt{m}$ when $m$ increases, thus faster than our theoretical error rate in $1/m^{1/4}$. The discrepancy between this two rates will be studied in future work, as well as the impact of phase quantization and additive noise on the phase-only sensing model.

\end{document}